\documentclass[12pt]{article}
\usepackage{amsmath, amssymb, amsthm, latexsym, epsfig}

\usepackage{geometry}

\usepackage{braket}

\DeclareSymbolFont{AMSb}{U}{msb}{m}{n}
\DeclareMathSymbol{\bbbc}{\mathalpha}{AMSb}{"43}
\DeclareMathSymbol{\bbbd}{\mathalpha}{AMSb}{"44}
\DeclareMathSymbol{\bbbn}{\mathalpha}{AMSb}{"4E}
\DeclareMathSymbol{\bbbo}{\mathalpha}{AMSb}{"4F}
\DeclareMathSymbol{\bbbp}{\mathalpha}{AMSb}{"50}
\DeclareMathSymbol{\bbbq}{\mathalpha}{AMSb}{"51}
\DeclareMathSymbol{\bbbr}{\mathalpha}{AMSb}{"52}
\DeclareMathSymbol{\bbbt}{\mathalpha}{AMSb}{"54}
\DeclareMathSymbol{\bbbz}{\mathalpha}{AMSb}{"5A}

\newtheorem{Proposition}{Proposition}

\newtheorem{Lemma}{Lemma}
\newtheorem{Corollary}{Corollary}

\DeclareMathSymbol{\B}{\mathalpha}{AMSb}{"42}
\DeclareMathSymbol{\C}{\mathalpha}{AMSb}{"43}
\DeclareMathSymbol{\D}{\mathalpha}{AMSb}{"44}
\DeclareMathSymbol{\E}{\mathalpha}{AMSb}{"45}
\DeclareMathSymbol{\F}{\mathalpha}{AMSb}{"46}
\DeclareMathSymbol{\N}{\mathalpha}{AMSb}{"4E}
\DeclareMathSymbol{\Q}{\mathalpha}{AMSb} {"51}
\DeclareMathSymbol{\R}{\mathalpha}{AMSb}{"52}
\DeclareMathSymbol{\T}{\mathalpha}{AMSb}{"54}
\DeclareMathSymbol{\Z}{\mathalpha}{AMSb}{"5A}

\begin{document}
\sffamily
\title{The Coxeter symmetries of high-level general sequential computation in the invertible-digital and quantum domains} \vspace{1em} 
\author{H.E. Bez\\\\  Department of Computer Science \\ Loughborough University \\ Leicestershire LE11 3TU
\\e-mail: 
\\
h.e.bez@lboro.ac.uk   
\\
\;helmut.e.bez@gmail.com}
\date{}
\maketitle

\newpage
\begin{abstract}
\noindent  The article investigates high-level general invertible-sequential processing in the digital and quantum domains. In particular it is shown that:
\begin{enumerate}
\item invertible digital-sequential processes, constructed using a standard general-inversion procedure, and 
\item sequential quantum processes, 
\end{enumerate}
determine Coxeter groups. In each case the groups are seen to define all processes that may be constructed from the, given, elemental processes of the sequences. Explicit forms of the presentations of the Coxeter groups are given for all cases. The quantum processes are seen to define unitary representations of the associated Coxeter groups in tensor-product qubit
spaces.    
\end{abstract}
\section{Introduction and overview} \label{Sect1}
Using the well-known (see, for example, \cite{NielsonChuang}) general invertible representation 
\[
\hat{f}(x, y) = (x, y \oplus f(x)), \text{    where  } \oplus \text{  denotes the Boolean exclusive-or operator}
\]
of a Boolean function $f$, a process is defined to step-wise invert a general non-invertible sequence of Boolean functions $f_n \circ f_{n-1} \circ \cdots \circ f_1$. Each of the invertible component functions, which in the general case are denoted by $\phi_i$, of the invertible sequence $\phi_n \circ \phi_{n-1} \circ \cdots \circ \phi_1$ so constructed, is an involution - and a Coxeter group is defined by $\{\phi_1, \ldots, \phi_n\}$. For a $2$-step sequence, i.e., the case $n=2$, the group is shown to be isomorphic to the $8$-element dihedral group $D_8$. For $n \geq 3$ more general Coxeter groups are generated and their presentations are given for all cases.

\vspace{1em}
In the quantum computing case it is shown, for example, that the quantisation of a $2$-step digital sequence $g(f(x))$, where $f : \B^n \rightarrow \B^m$ and $g : \B^m \rightarrow \B^k$, leads to a linear representation of $D_8$ in the unitary group $U(\otimes^{n+m+k} \C^2)$ of the tensor-product space $\otimes^{n+m+k} \C^2$.  
\section{Non-invertible sequential digital processing} \label{Sect2}
The output of a sequential digital process depends on the initial input and the output of a sequence of previous computational steps. A $j$-step process of this type may therefore be specified as follows:
\begin{itemize}
\item pre-process state: $x \in \B^n$
\item post-process state: $(f_j \circ f_{j-1} \circ \cdots \circ f_2 \circ f_1)(x)$ for some functions $f_1, f_2, \ldots, f_j$.
\end{itemize}
In general such computations are not invertible; i.e., the pre-process state cannot be determined from the post-process state. It is implicit in the above that the functional composition $f_1 \circ f_2 \circ \cdots \circ f_k$ is well-defined. A coded representation of the above may be expressed as:
\[
\begin{array}{llllll}
x \leftarrow (x_1, \ldots, x_n)   \;\;\;\;\;\;\; \text{/* initialisation of the pre-process state */}
\\
\textbf{begin}
\\
 \;\; y_1 \leftarrow f_1(x)
\\
 \;\; y_2 \leftarrow f_2(y_1)
\\
\;\;\; \vdots
\\
\;\; y_j \leftarrow f_j(y_{j-1})
\\
\textbf{end}
\end{array}
\]
\section{Invertible sequential digital processing} \label{Sect3}
\subsection{The general-inversion process}  \label{ga2} 
For $f : \B^n \rightarrow \B^m$ and $g : \B^m \rightarrow \B^k$, the functional composition $g \circ f$ is defined, but as neither $f$ nor $g$ is invertible the function $g \circ f$ will not be invertible. The general-invertible representations of $f$ and $g$ are:
\begin{eqnarray*}
\hat{f}(x, y) &=& (x, y \oplus f(x))
\\
\hat{g}(y, p) &=& (y, p \oplus g(y))
\end{eqnarray*}  
where $\hat{f} : \B^{n+m} \rightarrow \B^{n+m}, \hat{g} : \B^{m+k} \rightarrow \B^{m+k}, x \in \B^n, y \in \B^m, p \in \B^k$. The functions $\hat{f}$ and $\hat{g}$ are invertible representations of $f$ and $g$ but $\hat{g} \circ \hat{f}$ is not defined and is not therefore representative of $2$-step invertible sequential digital processing. 
\subsection{A functional extension of the general-invertible process}  \label{ga4} 
Although $\hat{f}$ and $\hat{g}$ represent individual invertible digital computations we have seen above that their functional composition is, in general, undefined. For sequential processing, functions that combine under the functional composition operator $\circ$ are required. 

\vspace{1em}
One way of accomplishing an appropriate functional composition of $f$ and $g$ is to extend the $\hat{f}$ and $\hat{g}$ functions by Cartesian product. This process extends their domains and enables a functional composition, representative of invertible $2$-step digital sequential processes, to be defined. We denote the extended functions by $\phi$ and $\gamma$ and define them below. The generalisation to $n$-step invertible digital sequences is then clear.

\vspace{1em}
We define for $f$ and $g$ where $f : \B^n \rightarrow \B^m$ and $g : \B^m \rightarrow \B^k$ the invertible functions
\begin{eqnarray*}
\phi &=& \hat{f} \times e^k
\\
\gamma &=& e^n \times \hat{g} 
\end{eqnarray*}
where $e^q : \B^q \rightarrow \B^q$, $q = 1,2,  \ldots$, are identity functions. We have $\phi: \B^{n+m+k} \rightarrow \B^{n+m+k}$, $\gamma: \B^{n+m+k} \rightarrow \B^{n+m+k}$ with:
\begin{eqnarray*}
\phi(x,y, z) &=& \hat{f} \times e^k(x,y, z)
\\
&=& (\hat{f}(x,y), z)
\\
&=& (x, y \oplus f(x), z)
\end{eqnarray*}
and
\begin{eqnarray*}
\gamma(x,y,z) &=& e^n \times \hat{g}(x,y,z)
\\
&=& (x, \hat{g}(y,z))
\\
&=& (x, y, z \oplus g(y))
\end{eqnarray*}
from which we obtain
\begin{eqnarray*}
\gamma \circ \phi(x,y,z) &=& \gamma (x, y \oplus f(x), z)
\\
&=& (x, y \oplus f(x), z \oplus g(y \oplus f(x)))
\end{eqnarray*}
and
\[
\gamma \circ \phi(x, 0, 0) = (x, \; f(x), \; g(f(x))).
\]
Hence the result of the non-invertible sequential digital process, $x \rightarrow g(f(x))$, may be obtained, by invertible means, from the 3rd component of $\gamma \circ \phi(x, 0, 0)$.
\section{The Coxeter group generated by $2$-step invertible digital sequences} \label{Sect4}
Throughout we denote by $G(\B^p)$ the group, under functional composition, of invertible Boolean functions from $\B^p$ to $\B^p$. 
\subsection{Properties of the extended functions $\phi$ and $\gamma$}
\begin{Lemma} \label{FGL1}
The extended functions $\phi = \hat{f} \times e^k$ and $\gamma = e^n \times \hat{g}$ are both self-inverse.
\end{Lemma}
\begin{proof}
We have $\phi(x,y,z) = (x, y \oplus f(x), z)$ hence
\begin{eqnarray*}
\phi \circ \phi(x,y,z) &=& \phi(x, y \oplus f(x), z)
\\
&=& (x, (y \oplus f(x)) \oplus f(x), z)
\\
&=& (x, y, z)
\end{eqnarray*}
Similarly for $\gamma$.
\end{proof}
Hence $\phi \in G(\B^{n+m+k})$ and $\gamma \in G(\B^{n+m+k})$. It follows that the functional compositions $\phi \circ \gamma$ and $\gamma \circ \phi$ are also elements of the group $G(\B^{n+m+k})$; their properties, which are required to establish the main results of this article, are discussed below. 
\subsection{Properties of the composite functions $\gamma \circ \phi$ and $\phi \circ \gamma$} 
\begin{Lemma} \label{FGL2}
Each is the inverse of the other, i.e., $(\gamma \circ \phi)^{-1} = \phi \circ \gamma$ and $(\phi \circ \gamma)^{-1} = \gamma \circ \phi$.
\end{Lemma}
\begin{proof}
Follows trivially from the fact that $\phi$ and $\gamma$ are both self-inverse - Lemma \ref{FGL1}.  
\end{proof}
\begin{Lemma}  \label{FGL3}  
$\gamma \circ \phi(x,y,z) =  (x, y \oplus f(x), z \oplus g(y \oplus f(x) ))$.
\end{Lemma}
\begin{proof}
See section \ref{ga4} 
\end{proof}
\begin{Lemma} \label{FGL4} 
$\phi \circ \gamma \circ \phi(x,y,z) = (x, y, z \oplus g(y \oplus f(x)))$.
\end{Lemma}
\begin{proof}
From Lemma \ref{FGL3} we obtain:
\begin{eqnarray*}
\phi \circ \gamma \circ \phi(x,y,z)  &=& \phi (x, y \oplus f(x), z \oplus g(y \oplus f(x) ))
\\
&=& (x, y \oplus f(x) \oplus f(x), z \oplus g(y \oplus f(x)))
\\
&=& (x, y, z \oplus g(y \oplus f(x))).
\end{eqnarray*}
\end{proof}
\begin{Lemma} \label{FGL5} 
$(\gamma \circ \phi)^2 = (x, y, z \oplus g(y \oplus f(x)) \oplus g(y))$.
\end{Lemma}
\begin{proof} From Lemma \ref{FGL4} we have: 
\begin{eqnarray*}
(\gamma \circ \phi)^2(x,y,z) &=& \gamma(x, y, z \oplus g(y \oplus f(x))
\\
&=& (x, y, z \oplus g(y \oplus f(x)) \oplus g(y))
\end{eqnarray*}
\end{proof}
\begin{Lemma} \label{FGL6} 
$\phi \circ (\gamma \circ \phi)^2 = (x, y \oplus f(x),  [z \oplus g(y \oplus f(x))] \oplus g(y))$.
\end{Lemma}
\begin{proof}
From Lemma \ref{FGL5} we have
\begin{eqnarray*}
\phi \circ (\gamma \circ \phi)^2 (x,y, z) &=& \phi (x, y, z \oplus g(y \oplus f(x)) \oplus g(y))
\\
&=& (x, y \oplus f(x), z \oplus g(y \oplus f(x)) \oplus g(y)) 
\end{eqnarray*}
\end{proof}
\begin{Lemma} \label{FGL7} 
$(\gamma \circ \phi)^3(x, y, z)=(x, y \oplus f(x), z \oplus g(y))$.
\end{Lemma}
\begin{proof}
Applying $\gamma$ to the outcome of Lemma \ref{FGL6} gives: 
\begin{eqnarray*}
\gamma(x, y^+, z^+) &=& (x,  y^+, z^+ \oplus g(y^+))
\\
&=& (x, y \oplus f(x),  ( [z \oplus g(y \oplus f(x))] \oplus g(y)) \oplus g(y \oplus f(x)))
\\
&=& (x, y \oplus f(x),  ( [z \oplus g(y \oplus f(x)) \oplus g(y \oplus f(x)   ] \oplus g(y)) ))
\\
&=& (x, y \oplus f(x), z \oplus g(y)).
\end{eqnarray*}
\end{proof}
\begin{Lemma} \label{FGL8} 
$\phi \circ (\gamma \circ \phi)^3 = \gamma$.
\end{Lemma}
\begin{proof}
Pre-multiplying the result of Lemma \ref{FGL8} by $\phi$ we obtain:
\begin{eqnarray*}
\phi \circ (\gamma \circ \phi)^3(x, y, z) &=& \phi(x, y \oplus f(x), z \oplus g(y)) 
\\
&=& (x, y, z \oplus g(y))   
\\
&=& \gamma(x, y, z).
\end{eqnarray*}
\end{proof}
Denoting the identity function of $G(\B^{n+m+k})$ by $e$ we have:
\begin{Lemma} \label{FGL9} 
$(\gamma \circ \phi)^4 = e$, 
\end{Lemma}
\begin{proof}
From Lemma \ref{FGL8} we have $(\gamma \circ \phi)^4 = \gamma \circ (\phi \circ (\gamma \circ \phi)^3) = \gamma \circ \gamma = e$.
\end{proof}
Equivalently $\gamma \circ \phi$ is of order $4$ in the group $G(\B^{n+m+k})$. Given that the inverse of $\gamma \circ \phi$ is $\phi \circ \gamma$ we have: 
\begin{Corollary}
\begin{eqnarray*}
\phi \circ \gamma 	&=& (\gamma \circ \phi)^3,
\\
(\phi \circ \gamma)^2 &=& (\gamma \circ \phi)^2, 
\\
(\phi \circ \gamma)^3 &=&  \gamma \circ \phi,
\\
(\phi \circ \gamma)^4 &=& e. 
\end{eqnarray*}
\end{Corollary}
and, hence: 
\begin{Proposition} \label{Prop1}
The elements $\gamma \circ \phi$ and $\phi \circ \gamma$ are of order $4$ in the group  $G(\B^{n+m+k})$.
\end{Proposition}
\begin{Corollary}
$(\gamma \circ \phi)^2$ is self-inverse.
\end{Corollary}
\begin{proof}
Immediate from Lemma \ref{FGL9}.
\end{proof}
\begin{Corollary}
$C_{\gamma, \phi} = \{e, \gamma \circ \phi, (\gamma \circ \phi)^2, (\gamma \circ \phi)^3)\}$ is a cyclic group of order $4$.
\end{Corollary}
\begin{proof}
It follows from Lemma \ref{FGL9} that $(\gamma \circ \phi)^{-1} = (\gamma \circ \phi)^3, (\gamma \circ \phi)^2$ is self inverse, and $((\gamma \circ \phi)^3)^{-1} = \gamma \circ \phi$. Hence $C_{\gamma, \phi}$ is closed under inversion. It is clearly closed under multiplication.
\end{proof}
\begin{Corollary}
$C_{\phi, \gamma} = \{e, \phi \circ \gamma, (\phi \circ \gamma)^2, (\phi \circ \gamma)^3)\}$ is a cyclic group of order $4$.
\end{Corollary}
\begin{proof}
\end{proof}
\begin{Corollary}
$C_{\phi, \gamma} = C_{\gamma, \phi}$.
\end{Corollary}
\begin{proof}
We have
\begin{eqnarray*}
C_{\phi, \gamma} &=& \{e, \phi \circ \gamma, (\phi \circ \gamma)^2, (\phi \circ \gamma)^3)\}
\\
&=& \{e, (\gamma \circ \phi)^3, (\gamma \circ \phi)^2, \gamma \circ \phi\} = C_{\gamma, \phi}.
\end{eqnarray*}
\end{proof}
\subsection{The group generated by the functions $\phi$ and $\gamma$}
The cyclic group $C_{\gamma, \phi}$ is not the largest subgroup of $G(\B^{n+m+k})$ generated by $\phi$ and $\gamma$. We have:
\begin{Lemma} \label{FGL13}
The functions $\phi = (\hat{f} \times e)$ and $\gamma = (e \times \hat{g})$ generate the $8$-element sub-group  
\[
G_{2s} = \{e, \phi, \gamma, \phi \circ \gamma, (\phi \circ \gamma)^2, (\phi \circ \gamma)^3, \phi \circ (\gamma \circ \phi), \gamma \circ (\phi \circ \gamma)\}
\]
of $G(\B^{n+m+k})$.
\end{Lemma}
\begin{proof} 
The functions $\phi$ and $\gamma$ clearly generate the elements of $G_{2s}$. The closure of $G_{\phi, \gamma}$ under both inversion and product, are easily verified from the properties of $\phi, \gamma$, and their products, established above.
\end{proof}
\begin{Lemma} \label{FGL14}
The subgroup $G_{2s} \subset G(\B^{n+m+k})$, defined in Lemma \ref{FGL13}, is isomorphic to the dihedral group $D_8$.
\end{Lemma}
\begin{proof} 
Following, for example, Ledermann \cite{Ledermann} the dihedral group $D_8$ may be defined by:
\[
\{e, A, A^2, A^3, B, AB, A^2B, A^3B\}
\]
where $A$ and $B$ are such that 
\[
A^4 = B^2 = (AB)^2 = e.
\]
 With the identifications 
 \[
 A = \phi \circ \gamma \text{  and  } B = \gamma
 \]
 the conditions $A^4 = B^2 = (AB)^2 = e$ are satisfied and we have:  
\begin{eqnarray*}
\{e, A, A^2, A^3, B, AB, A^2B, A^3B\} &=& \{e, \phi \circ \gamma,  (\phi \circ \gamma)^2, (\phi \circ \gamma)^3, \gamma, \phi, \phi \circ (\gamma \circ \phi), \gamma \circ (\phi \circ \gamma)\}
\\
&=& G_{2s}.
\end{eqnarray*}
\end{proof}
We note that the group $G_{2s}$ may also be expressed as
\[
G_{2s} = \{e, \phi, \gamma, \gamma \circ \phi, (\gamma \circ \phi)^2, (\gamma \circ \phi)^3, \phi \circ (\gamma \circ \phi), \gamma \circ (\phi \circ \gamma)\}
\]
and, from the digital-computational point of view, the elements $\phi, \gamma$ and $\gamma \circ \phi$ are sufficient to represent:
\begin{enumerate}
\item the non-invertible computations defined by $f$ and $g$, and 
\item the $2$-step sequence, or functional composition, $g \circ f$
\end{enumerate}
by invertible means. However $\{\phi, \gamma, \gamma \circ \phi\}$ is not a subgroup of $G_{2s}$, and $\gamma \circ \phi$ isn't the only invertible digital sequence that may be constructed from the elements $\gamma$ and $\phi$. The group $G_{2s}$ is clearly representative of the complete set of invertible digital sequences (or invertible \lq programs') that may be constructed from the invertible components $\phi$ and $\gamma$. 
\subsection{Coxeter groups}
In terms of group presentation the group $G_{2s}$, hence also $D_8$, may be defined by:
\[
G_{2s} \equiv \langle A, B|  A^4 = B^2 = (AB)^2 = e \rangle
\]
i.e., the set of strings over $\{A, B\}$ consistent with the relations $A^4 = B^2 = (AB)^2 = e$.
The Coxeter groups have been widely studied following their introduction by H. S. M. Coxeter in 1934 \cite{Coxeter}. They have applications in many areas of mathematics and are defined by presentation. We have for $n  > 1$ in $\N$ and $1 < i, j \leq n$ the Coxeter group generated by the elements $s_1, \ldots, s_n$ has the presentation:
\[
\langle s_1, \ldots s_n  | (s_i s_j)^{m_{ij}} = e\rangle
\]
where, for all $i,j \in \{1, \ldots n\}$, we have $m_{ij} \in \N$ where $m_{ii} = 1$ and for $i \neq j$ we have $m_{ij} > 1$ and $m_{ij} = m_{ji}$

\vspace{1em}
The dihedral groups $D_{2n}$ are special cases of Coxeter groups with $n = 2$. As a Coxeter group $G_{2s}$ has the presentation:
\begin{eqnarray*}
G_{2s} &\equiv& \langle s_1, s_2| s_1s_1 = e,  s_2s_2 = e, (s_1 s_2)^4 = e \rangle,
\\
&=& \langle \phi, \gamma| \phi^2 = e, \; \gamma^2 = e, \; (\phi \circ \gamma)^4 = e \rangle,
\end{eqnarray*}
i.e., $m_{12} = m_{21} = 4$.
\section{The Coxeter groups of general $n$-step invertible digital sequences}  \label{Sect5}
\subsection{The case of $3$ sequential steps}
In the case of $3$-step invertible digital computations we have, from the general invertible process,  the $3$ independent involutions, $\phi, \gamma, \rho$, defined by: 
\begin{eqnarray*}
\phi(x, y, z, t) &=& (x, y \oplus f(x), z, t) 
\\
\gamma(x, y, z, t) &=& (x, y, z \oplus g(y), t)
\\
\rho(x, y, z, t) &=& (x, y, z, t \oplus h(z)) 
\end{eqnarray*}
from the non-invertible digital functions $f, g$ and $h$ for which the functional product $h(g(f(x)))$ is well-defined. We obtain
\begin{eqnarray*}
\rho \circ \gamma \circ \phi(x, y, z, t) &=& \rho \circ \gamma (x, y \oplus f(x), z, t)
\\
&=& \rho(x, y \oplus f(x), z \oplus g(y \oplus f(x)), t)
\\
&=& (x, y \oplus f(x), z \oplus g(y \oplus f(x)), t \oplus h(z \oplus g(y \oplus f(x))) ).
\end{eqnarray*}
and have
\[
\rho \circ \gamma \circ \phi(x, 0, 0, 0) = (x, \; f(x), \; g(f(x)), \; h(g(f(x))))
\]
i.e., the outcome of the digital sequence $h(g(f(x)))$, available in the last component of the above, is computed by invertible means. 

\vspace{1em}
It is straightforward to show that the functions $\phi, \gamma, \rho$ satisfy the relationships:
\begin{eqnarray*}
\phi^2 = \gamma^2 = \rho^2 = e, \;  (\phi \circ \gamma)^4 = e, \;  (\gamma \circ \rho)^4 = e, \; (\phi \circ \rho)^2 = e
\end{eqnarray*}
from which it follows that the group generated by $\phi, \gamma, \rho$ is the Coxeter group with presentation:
\[
G_{3s} \equiv \langle \phi, \gamma, \rho | \phi^2 = e, \gamma^2 = e, \rho^2 = e, (\phi \circ \gamma)^4 = e, (\gamma \circ \rho)^4 = e, (\phi \circ \rho)^2 = e   \rangle.
\]
\subsection{The case of $n$ sequential steps}
For the $n$-step non invertible digital sequence $x \rightarrow f_n \circ f_{n-1} \circ \cdots \circ f_2 \circ f_1(x)$ we define, from the general invertible construction, the invertible functions
$\hat{f}_1, \hat{f}_2, \ldots \hat{f}_n$ and their extensions \newline 
\[
\phi_1, \phi_2, \ldots, \phi_n
\]
such that the product 
\[
\phi_n \circ \phi_{n-1} \circ \cdots \circ \phi_1
\]
is a well-defined functional product for the computing the non-invertible sequence 
\[
f_n \circ f_{n-1} \circ \cdots \circ f_2 \circ f_1
\]
by invertible means. The group $G_{ns}$ for the $n$-step invertible digital sequence is the Coxeter group defined by the presentation: 
\[
G_{ns} \equiv \langle \phi_1, \ldots , \phi_n   | \phi_1^2 = \phi_2^2 = \cdots = \phi_n^2 = e; \; (\phi_k \circ \phi_{k+1})^4 = e; \; \text{else  } (\phi_p \circ \phi_q)^2 = e    \rangle
\]
i.e., when $q \neq p+1$ we have $(\phi_p \circ \phi_q)^2 = e$. Equivalently the presentation of $G_{ns}$ may be expressed as:  

\[
\langle \phi_1, \ldots , \phi_n   | R_1 \cup R_2 \cup R_3\rangle
\]
where
\begin{eqnarray*}
R_1 &=& \phi_1^2 = \phi_2^2 = \cdots = \phi_n^2 = e
\\
R_2 &=& (\phi_k \circ \phi_{k+1})^4 = e 
\\
R_3 &=& (\phi_p \circ \phi_q)^2 = e \;\; \text{when} \;\; |q - p| \neq 1.
\end{eqnarray*}
\section{Quantum sequential processing} \label{Sect6}
\subsection{$\B^{n+m+k}$ as  a $G_{2s}$-space}
From sections \ref{Sect3} and \ref{Sect4} we have an action of $G_{2s}$, equvalently $D_8$, on $\B^{n+m+k}$ defined by:
\begin{eqnarray*}
\phi(x, y, z) 						&=& (x, y \oplus f(x), z) 
\\
\gamma(x,y,z) 						&=& (x, y, z \oplus g(y))
\\
\gamma \circ \phi(x, y, z) 				&=& (x, y \oplus f(x), z \oplus g(y \oplus f(x)) 
\\
(\gamma \circ \phi)^2(x, y, z) 			&=& (x, y, z \oplus g(y \oplus f(x)) \oplus g(y) ) 
\\
(\gamma \circ \phi)^3(x, y, z) 			&=& (x, y \oplus f(x), z \oplus g(y))
\\
\phi \circ (\gamma \circ \phi)(x, y, z) 		&=& (x, y, z \oplus g(y \oplus f(x)))
\\
\gamma \circ (\phi \circ \gamma)(x, y, z) 	&=& (x, y \oplus f(x), (z \oplus g(y)) \oplus g(y \oplus f(x))).
\end{eqnarray*}
From this we define corresponding actions on the basis elements: 
\[
B_{\otimes^{n+m+k}\C^2} = \{|x\rangle |y\rangle |z\rangle: x \in \B^n, y \in \B^m, z \in \B^k\}
\]
of $\otimes^{n+m+k}\C^2$ by 
\begin{eqnarray*}
U_\phi |x\rangle |y\rangle |z \rangle &=& |x\rangle |y \oplus f(x)\rangle |z\rangle
\\
U_\gamma |x\rangle |y\rangle |z \rangle &=& |x\rangle |y\rangle |z \oplus g(y)\rangle
\\
U_{\gamma \circ \phi}|x\rangle|y\rangle|z\rangle &=& |x\rangle|y \oplus f(x) \rangle |z \oplus g(y \oplus f(x)) \rangle
\\
&\vdots&
\\
U_{\gamma \circ (\phi \circ \gamma)}|x\rangle |y\rangle |z\rangle &=& |x\rangle |y \oplus f(x)\rangle |(z \oplus g(y)) \oplus g(y \oplus f(x))\rangle
\end{eqnarray*}
the linear extensions of which, from $B_{\otimes^{n+m+k}\C^2}$  to the whole of $\otimes^{n+m+k}\C^2$, determine unitary operators. 

\vspace{1em}
We note that:
\[
U_{\gamma \circ \phi}|x\rangle|0\rangle|0\rangle = |x\rangle|f(x) \rangle | g(f(x)) \rangle
\]
which is a quantum implementation of the digital sequence $g(f(x))$. The value of $g(f(x))$ being obtained from $U_{\gamma \circ \phi}|x\rangle|0\rangle|0\rangle$ by measurement of the $k$ qubits $|g(f(x))\rangle$. 

\vspace{1em}
From the above we also have
\[
U_{\gamma \circ \phi} = U_{\gamma} U_{\phi}
\]
and
\begin{Proposition}
For $\psi \in G_{2s}$ the mapping $\psi \rightarrow U_\psi$ is a unitary representation of the group $G_{2s}$ in the unitary group of the tensor-product space $\otimes^{n + m + k} \C^2$.
\end{Proposition}
The quantisation of $n$-step digital sequences leads to similar conclusions.
\section{Conclusion} \label{Sect7}
Coxeter groups have been shown to occur in both high-level, $n$-step, invertible digital sequential processing and in a quantum representation of sequential processing. The groups, in both cases, have been shown to be representative of the set of all invertible processes that may be constructed from the given set of processing steps. 

\vspace{1em}
It has also been shown that, in the quantum case, particular unitary representations of the Coxeter groups occur; the general nature and properties of these representations is not known to the author. 

\end{document}